\documentclass[runningheads]{llncs}

\usepackage[T1]{fontenc}
\usepackage[utf8]{inputenc}
\usepackage{amsmath,amssymb}
\usepackage{booktabs,tabularx,array}
\usepackage{graphicx}
\usepackage{float}
\usepackage{placeins}
\usepackage{xcolor}
\usepackage{tikz}
\usetikzlibrary{arrows.meta,positioning,fit,calc,shapes.geometric}
\usepackage{algorithm}
\usepackage{algpseudocode}
\usepackage{url}
\usepackage[hidelinks]{hyperref}

\hypersetup{
  colorlinks=true,
  linkcolor=blue!55!black,
  citecolor=blue!55!black,
  urlcolor=blue!55!black
}

\newcommand{\QuPort}{\textsc{QuPort}}

\newcommand{\cut}{\operatorname{cut}}
\newcommand{\load}{\operatorname{load}}

\begin{document}

\title{QuPort: Topology-, Port-, and Congestion-Aware Compilation for Modular Multi-QPU Quantum Systems}
\titlerunning{QuPort Multi-QPU Compilation}
\author{
Soumyadip Sarkar\inst{1} \and
Subhasree Bhattacharjee\inst{2}
}
\institute{
Independent Researcher
\and
Department of Computer Application, Narula Institute of Technology, Kolkata, India
}
\authorrunning{S. Sarkar \and S. Bhattacharjee}
\maketitle

\begin{abstract}
Modular quantum processors require a compiler to reason about two resources at the same time: local device connectivity and communication across QPUs. A mapping that is acceptable on a single coupling graph may be unsuitable for a modular machine if it creates excessive cross-QPU traffic, concentrates that traffic on a small number of interconnect links, or assigns many boundary qubits to a QPU with few communication ports. This paper presents \QuPort, a Python and Qiskit-based compilation framework that studies this setting through an explicit three-level model: a weighted logical interaction graph, a directed physical coupling map, and an undirected QPU-level interconnect graph. The main partitioning method, TPCCAP, optimizes the implemented objective formed by weighted cut distance, communication-port overflow, and routed link-load congestion. The framework also includes heavy-edge clustering, balanced greedy partitioning, simulated-annealing refinement, communication-port-aware layout, extraction of remote two-qubit operations, local-only routing of per-QPU circuits, and topology-aware schedule estimation. The model is a compiler-level abstraction. It does not claim a calibrated hardware runtime or an implementation of a physical remote-gate protocol.
\let\thefootnote\relax\footnote{The code can be found at: \url{https://github.com/neuralsorcerer/quport}}
\end{abstract}

\section{Introduction}

Quantum compilation translates an abstract circuit into instructions that respect the basis gates and connectivity of a target device. On current monolithic devices, this problem is usually expressed through basis translation, initial layout, and routing over a coupling map. Qiskit's transpiler follows this model, and its \texttt{CouplingMap} represents fixed directed couplings between physical qubits \cite{qiskit_transpile_docs,qiskit_coupling_docs}. If two logical qubits that must interact are not adjacent under the target connectivity, the routing stage may introduce SWAPs or related transformations.

A modular quantum processor changes the meaning of a nonlocal interaction. If two operands belong to different QPUs, the operation is not simply a longer local gate on the same device. Its implementation depends on the interconnect and the physical platform. Proposed and demonstrated distributed systems may use photonic links, entanglement generation, measurement, classical communication, state transfer, or gate teleportation. Recent trapped-ion and superconducting modular experiments support the relevance of this direction, but they also indicate that inter-module communication remains a scarce resource rather than a free extension of local connectivity \cite{main2025distributed,mollenhauer2025superconducting}.

\QuPort{} addresses the compiler problem that appears before a platform-specific remote-gate implementation is chosen. Given a circuit and a modular architecture description, it extracts two-qubit interaction weights, assigns logical qubits to QPUs, selects communication-port placements, and then follows either a global or a distributed compilation path. In global mode, a single directed physical coupling map is passed to Qiskit. In distributed mode, cross-QPU two-qubit operations are extracted as explicit remote events, while each local circuit is routed only on the intra-QPU coupling map. This distinction prevents inter-QPU communication from being hidden inside ordinary monolithic routing.

The paper makes three technical points. First, modular compilation benefits from separating the logical interaction graph, the physical coupling map, and the QPU interconnect graph. Second, the TPCCAP objective implemented in \QuPort{} captures three compiler-level sources of modular cost: QPU distance, communication-port pressure, and routed congestion. Third, the distributed compilation path produces an intermediate representation in which local circuits and remote events are separated, enabling later replacement of the abstract remote-event layer by a hardware-specific protocol.

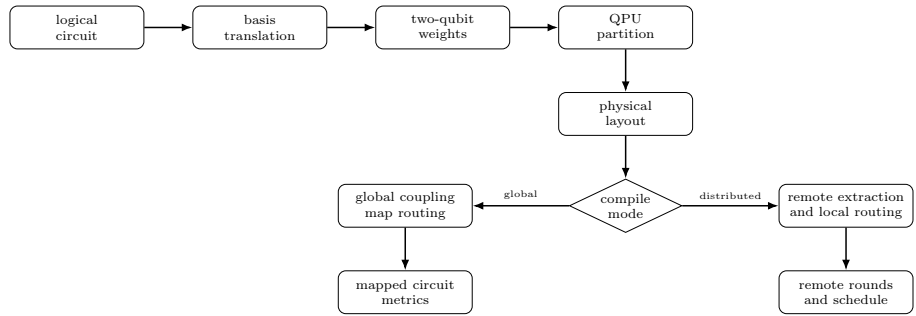
\begin{figure}[!htbp]
\centering
\resizebox{0.98\textwidth}{!}{%
\begin{tikzpicture}[
  node distance=8mm and 9mm,
  stage/.style={draw, rounded corners, align=center, minimum height=8mm, minimum width=25mm, font=\scriptsize},
  decision/.style={draw, diamond, aspect=2.1, align=center, inner sep=1pt, font=\scriptsize},
  arr/.style={-{Latex[length=2mm]}, thick}
]
\node[stage] (input) {logical\\circuit};
\node[stage, right=of input] (basis) {basis\\translation};
\node[stage, right=of basis] (weights) {two-qubit\\weights};
\node[stage, right=of weights] (partition) {QPU\\partition};
\node[stage, below=of partition] (layout) {physical\\layout};
\node[decision, below=of layout] (mode) {compile\\mode};
\node[stage, left=18mm of mode] (global) {global coupling\\map routing};
\node[stage, right=18mm of mode] (distributed) {remote extraction\\and local routing};
\node[stage, below=of global] (metrics) {mapped circuit\\metrics};
\node[stage, below=of distributed] (schedule) {remote rounds\\and schedule};
\draw[arr] (input) -- (basis);
\draw[arr] (basis) -- (weights);
\draw[arr] (weights) -- (partition);
\draw[arr] (partition) -- (layout);
\draw[arr] (layout) -- (mode);
\draw[arr] (mode) -- node[above, font=\tiny] {global} (global);
\draw[arr] (mode) -- node[above, font=\tiny] {distributed} (distributed);
\draw[arr] (global) -- (metrics);
\draw[arr] (distributed) -- (schedule);
\end{tikzpicture}%
}
\caption{Compilation paths represented in \QuPort. The global path invokes Qiskit over one directed physical coupling map. The distributed path preserves cross-QPU gates as remote events and routes only inside each QPU.}
\label{fig:flow}
\end{figure}

\section{Related Work}

\paragraph{Single-device routing.}
Qubit routing for NISQ processors is often formulated as a layout and SWAP-insertion problem over a limited coupling graph. SABRE introduced a bidirectional heuristic for initial mapping and routing that remains influential in practical transpilation workflows \cite{li2019sabre}. Retargetable compilers such as \textsc{tket} combine circuit rewriting and hardware-aware routing for heterogeneous NISQ devices \cite{sivarajah2020tket}. \QuPort{} uses Qiskit as the underlying circuit and transpilation ecosystem, but it adds an explicit QPU-level interconnect model that is separate from the physical coupling map.

\paragraph{Distributed quantum compilation.}
Distributed quantum computation has been studied as a path toward scaling beyond a single processor \cite{cirac1999distributed,monroe2014modular}. Compiler work in this area treats nonlocal gates as resources that must be assigned, exposed, and scheduled rather than ordinary nearest-neighbor gates. Ferrari et al. studied compiler design for distributed quantum computing and later presented a modular compilation framework that includes network-aware considerations \cite{ferrari2021compiler,ferrari2023modular}. Davarzani et al. considered hierarchical construction of distributed quantum systems with attention to inter-subsystem communication \cite{davarzani2022hierarchical}. Recent survey work also describes distributed quantum computing as a networked model in which computation and communication resources must be considered together \cite{caleffi2024survey}. \QuPort{} is aligned with these goals, while using a compact Python/Qiskit implementation centered on partitioning, port placement, remote-event extraction, and schedule estimation.

\paragraph{Classical graph partitioning.}
The logical assignment problem resembles weighted graph partitioning with capacity constraints. Multilevel partitioning methods, such as those of Karypis and Kumar, provide strong general-purpose methods for irregular graphs \cite{karypis1998multilevel}. \QuPort{} does not depend on an external partitioner. It implements transparent heuristics that are easy to inspect and modify: heavy-edge clustering, balanced greedy placement, TPCCAP, and TPCCAP-SA. This design supports reproducible compiler studies, but it should not be interpreted as a claim that these heuristics dominate specialized graph or hypergraph partitioning packages.

\section{System Model}

Let the input circuit after basis translation be
\begin{equation}
C=(Q_L,\mathcal{G}),\qquad Q_L=\{0,1,\ldots,n-1\},
\end{equation}
where \(Q_L\) is the logical-qubit set and \(\mathcal{G}\) is the ordered gate list. \QuPort{} extracts a weighted undirected logical interaction graph
\begin{equation}
G_L=(V_L,E_L,w),\qquad V_L=Q_L.
\end{equation}
For every two-qubit instruction on logical qubits \(i\) and \(j\), the canonical edge weight is incremented:
\begin{equation}
w_{ij}\leftarrow w_{ij}+1,\qquad i<j.
\label{eq:edge-count}
\end{equation}
For temporal weighting, the \(t\)-th two-qubit operation contributes \(\gamma^t\), where \(\gamma\in(0,1]\). Thus
\begin{equation}
W_{ij}=\sum_{t\in T_{ij}}\gamma^t,
\end{equation}
where \(T_{ij}\) is the set of two-qubit interaction times for the pair \((i,j)\). When \(\gamma=1\), this reduces to the ordinary count in Eq.~\eqref{eq:edge-count}.\\

The modular target has \(N\) QPUs. In the implemented configuration, each QPU has \(C\) compute qubits and \(P\) communication qubits. The block size is
\begin{equation}
B=C+P,
\end{equation}
and QPU \(q\) owns physical indices
\begin{equation}
\{qB,qB+1,\ldots,qB+B-1\}.
\end{equation}
The implemented physical-to-QPU map is
\begin{equation}
\operatorname{qpu}(p)=\left\lfloor \frac{p}{B}\right\rfloor.
\label{eq:qpu-of-phys}
\end{equation}
The local physical graph inside a QPU may be a clique, line, ring, or two-dimensional grid. Symmetric physical links are encoded as two directed edges because Qiskit's coupling map is directed \cite{qiskit_coupling_docs}.\\

Inter-QPU connectivity is modeled by an undirected graph
\begin{equation}
G_Q=(V_Q,E_Q),\qquad V_Q=\{0,\ldots,N-1\}.
\end{equation}
\QuPort{} includes switch, mesh, ring, degree-bounded, Clos-style, and fat-tree-style QPU graph abstractions. These are compiler interconnect abstractions, not calibrated descriptions of a particular hardware installation.\\

A logical-to-QPU partition is a map
\begin{equation}
\pi:Q_L\rightarrow V_Q.
\end{equation}
A physical layout is an injective map
\begin{equation}
\ell:Q_L\rightarrow Q_P,
\end{equation}
where \(Q_P\) is the set of physical qubits. Feasibility requires
\begin{equation}
\operatorname{qpu}(\ell(i))=\pi(i)
\end{equation}
for each logical qubit \(i\). A two-qubit operation on \((i,j)\) is local if \(\pi(i)=\pi(j)\) and remote otherwise.

\begin{figure}[!htbp]
\centering
\resizebox{0.96\textwidth}{!}{%
\begin{tikzpicture}[
  qubit/.style={circle, draw, minimum size=5.8mm, inner sep=0pt, font=\scriptsize},
  box/.style={draw, rounded corners, minimum width=31mm, minimum height=25mm, align=center, font=\scriptsize},
  edge/.style={thick},
  arr/.style={-{Latex[length=1.8mm]}, thick}
]
\node[font=\small] at (0,2.05) {logical interaction graph};
\node[qubit] (l0) at (-0.9,1.0) {$0$};
\node[qubit] (l1) at (0.9,1.0) {$1$};
\node[qubit] (l2) at (-0.9,-0.2) {$2$};
\node[qubit] (l3) at (0.9,-0.2) {$3$};
\draw[edge] (l0) -- node[above, font=\tiny] {$w_{01}$} (l1);
\draw[edge] (l0) -- (l2);
\draw[edge] (l1) -- node[right, font=\tiny] {$w_{13}$} (l3);
\draw[edge, dashed] (l2) -- (l3);

\node[font=\small] at (5.1,2.05) {physical coupling map};
\node[box] (qpa) at (3.9,0.35) {};
\node[box] (qpb) at (6.3,0.35) {};
\node[font=\scriptsize] at (3.9,1.25) {QPU $a$};
\node[font=\scriptsize] at (6.3,1.25) {QPU $b$};
\node[qubit] (pa) at (3.25,0.65) {};
\node[qubit] (pb) at (3.75,0.65) {};
\node[qubit] (pc) at (4.25,0.65) {};
\node[qubit, fill=gray!15] (pca) at (3.75,-0.35) {$c$};
\node[qubit] (pd) at (5.65,0.65) {};
\node[qubit] (pe) at (6.15,0.65) {};
\node[qubit] (pf) at (6.65,0.65) {};
\node[qubit, fill=gray!15] (pcb) at (6.15,-0.35) {$c$};
\draw[edge] (pa) -- (pb) -- (pc);
\draw[edge] (pb) -- (pca);
\draw[edge] (pd) -- (pe) -- (pf);
\draw[edge] (pe) -- (pcb);
\draw[arr] (pca) -- (pcb);
\draw[arr] (pcb) to[bend left=16] (pca);
\node[font=\tiny] at (5.05,-0.95) {communication qubits marked $c$};

\node[font=\small] at (9.4,2.05) {QPU interconnect graph};
\node[qubit] (q0) at (8.6,0.9) {$0$};
\node[qubit] (q1) at (10.2,0.9) {$1$};
\node[qubit] (q2) at (8.6,-0.4) {$2$};
\node[qubit] (q3) at (10.2,-0.4) {$3$};
\draw[edge] (q0) -- (q1);
\draw[edge] (q0) -- (q2);
\draw[edge] (q1) -- (q3);
\draw[edge] (q2) -- (q3);
\draw[edge, dashed] (q0) -- (q3);
\end{tikzpicture}%
}
\caption{Three graph views used by \QuPort. The logical graph stores circuit interaction weights, the physical coupling map is the directed graph passed to Qiskit, and the QPU graph is used for hop distance, traffic, congestion, and scheduling. The figure is illustrative.}
\label{fig:three-graphs}
\end{figure}
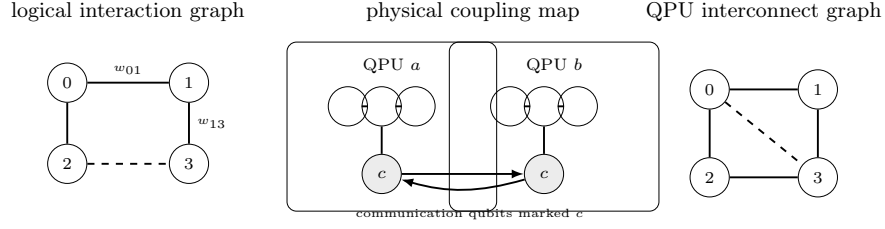

\section{Partitioning Objective}

The basic remote-interaction cut of a partition is
\begin{equation}
\cut(\pi)=\sum_{(i,j)\in E_L}w_{ij}\,\mathbf{1}[\pi(i)\neq\pi(j)].
\end{equation}
Cut weight alone does not distinguish a remote interaction across one QPU-network hop from one across several hops. It also does not account for communication-port scarcity or traffic concentration. \QuPort{} therefore computes a symmetric QPU traffic matrix
\begin{equation}
T_{ab}=\sum_{(i,j)\in E_L}w_{ij}\,\mathbf{1}[\{\pi(i),\pi(j)\}=\{a,b\}],\qquad a\neq b,
\end{equation}
with \(T_{aa}=0\). Let \(d(a,b)\) be the shortest-path distance in \(G_Q\), and let \(b_q\) be the number of boundary logical qubits assigned to QPU \(q\), where a boundary qubit has at least one interaction with a logical qubit assigned to another QPU.\\

The implemented TPCCAP objective is
\begin{equation}
J(\pi)=
\alpha \sum_{(i,j)\in E_L,\pi(i)\neq\pi(j)} w_{ij}d(\pi(i),\pi(j))
+\beta \sum_{q\in V_Q}\max(0,b_q-P)^2
+\eta \sum_{e\in E_Q}L_e^2.
\label{eq:tpccap-objective}
\end{equation}
The first term is weighted cut distance. The second term penalizes boundary-qubit count beyond the number of communication ports. The third term penalizes routed congestion, where \(L_e\) is the traffic load assigned to QPU-network edge \(e\). Traffic can be routed on one shortest path or split across equal-cost shortest paths. If traffic exists between disconnected QPU pairs, the implementation assigns a large penalty rather than treating the traffic as routable.\\

Equation~\eqref{eq:tpccap-objective} is not a physical fidelity model. It is the compiler objective implemented for architecture-aware partitioning. It does not include calibrated gate error, crosstalk, memory lifetime, entanglement fidelity, or queueing delay from a hardware control stack.

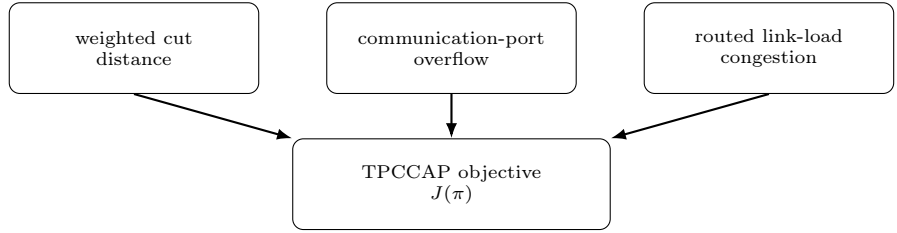
\begin{figure}[htbp]
\centering
\begin{tikzpicture}[
  term/.style={draw, rounded corners, align=center, minimum width=33mm, minimum height=12mm, font=\scriptsize},
  arr/.style={-{Latex[length=2mm]}, thick}
]
\node[term] (dist) at (0,0) {weighted cut\\distance};
\node[term] (port) at (4.2,0) {communication-port\\overflow};
\node[term] (cong) at (8.4,0) {routed link-load\\congestion};
\node[term, minimum width=42mm] (obj) at (4.2,-1.8) {TPCCAP objective\\$J(\pi)$};
\draw[arr] (dist.south) -- (obj.north west);
\draw[arr] (port.south) -- (obj.north);
\draw[arr] (cong.south) -- (obj.north east);
\end{tikzpicture}
\caption{Terms optimized by TPCCAP. Each term is computed from the logical interaction graph, the QPU partition, and the QPU-level interconnect graph.}
\label{fig:tpccap-terms}
\end{figure}

\section{Algorithms}

This section describes the algorithms present in \QuPort. The descriptions use mathematical names for clarity, but each algorithm corresponds to a concrete component of the implementation at the referenced repository snapshot \cite{quport_repo}.

\subsection{Heavy-edge clustering}

Heavy-edge clustering constructs capacity-bounded clusters before assigning them to QPUs. It sorts interaction edges by decreasing weight and merges the components incident to an edge when the merged component size remains at most \(K=C+P\). Clusters are then placed by first-fit decreasing bin packing. The method is simple and interpretable: high-weight logical pairs are kept together whenever capacity allows.

\begin{algorithm}[htbp]
\caption{Heavy-edge clustering partition}
\label{alg:heavy}
\begin{algorithmic}[1]
\Require Logical qubits \(Q_L\), edge weights \(w\), QPU count \(N\), capacity \(K\)
\Ensure Partition \(\pi\)
\State Initialize a disjoint-set structure with singleton components.
\State Sort edges \((i,j)\in E_L\) by decreasing \(w_{ij}\), using deterministic tie breaks.
\For{each edge \((i,j)\) in sorted order}
  \State Let \(A\) and \(B\) be the current components containing \(i\) and \(j\).
  \If{\(A\neq B\) and \(|A|+|B|\leq K\)}
    \State Merge \(A\) and \(B\).
  \EndIf
\EndFor
\State Sort components by decreasing size.
\State Place each component into the first QPU with sufficient remaining capacity.
\State Place any unplaced singleton qubits into remaining capacity.
\end{algorithmic}
\end{algorithm}

\subsection{Balanced greedy partitioning}

The balanced greedy strategy assigns logical qubits one at a time in descending weighted degree. For a candidate placement of logical qubit \(v\) on QPU \(q\), the score is
\begin{equation}
S(v,q)=\sum_{u:\pi(u)=q}w_{uv}-\lambda \frac{\load(q)}{K}.
\end{equation}
The first term rewards placing \(v\) with already placed neighbors. The second term discourages early overloading of a QPU. After the greedy assignment, the implementation performs local move refinement: a qubit may move to another non-full QPU when the move decreases cut weight.

\begin{algorithm}[htbp]
\caption{Balanced greedy partition with local refinement}
\label{alg:balanced}
\begin{algorithmic}[1]
\Require Edge weights \(w\), QPU count \(N\), capacity \(K\), balance weight \(\lambda\)
\Ensure Partition \(\pi\)
\State Order logical qubits by descending weighted degree.
\For{each logical qubit \(v\) in this order}
  \For{each QPU \(q\) with \(\load(q)<K\)}
    \State Compute \(S(v,q)=\sum_{u:\pi(u)=q}w_{uv}-\lambda\load(q)/K\).
  \EndFor
  \State Assign \(v\) to the feasible QPU with maximum score.
\EndFor
\Repeat
  \State Scan logical qubits in a randomized order.
  \State Move a qubit only if the move reduces cut weight and preserves capacity.
\Until{no improving move is found or the pass limit is reached}
\end{algorithmic}
\end{algorithm}

\subsection{TPCCAP local search}

TPCCAP starts from the balanced greedy partition and optimizes the objective in Eq.~\eqref{eq:tpccap-objective}. For a logical qubit \(v\), the implementation constructs a small candidate set of destination QPUs based on affinity to \(v\)'s neighbors, together with the current QPU. A move is tested by temporarily changing the assignment, recomputing the objective, and accepting the best reducing move that preserves capacity.

\begin{algorithm}[htbp]
\caption{TPCCAP local search}
\label{alg:tpccap}
\begin{algorithmic}[1]
\Require Initial feasible partition \(\pi\), weights \(w\), QPU shortest paths, port count \(P\)
\Ensure Improved partition \(\pi\)
\State Compute current objective \(J(\pi)\).
\Repeat
  \State Set \(\mathrm{changed}\leftarrow\mathrm{false}\).
  \For{each logical qubit \(v\) in randomized order}
    \State Build candidate QPUs from neighbor-affinity scores and include \(\pi(v)\).
    \State Let \(q^*=\pi(v)\) and \(J^*=J(\pi)\).
    \For{each candidate QPU \(q\neq\pi(v)\) with remaining capacity}
      \State Temporarily move \(v\) to \(q\) and evaluate \(J\).
      \If{the objective is lower than \(J^*\)}
        \State Store \(q\) as the best destination.
      \EndIf
    \EndFor
    \If{a better destination was found}
      \State Apply the move and update \(J(\pi)\).
      \State Set \(\mathrm{changed}\leftarrow\mathrm{true}\).
    \EndIf
  \EndFor
\Until{\(\mathrm{changed}=\mathrm{false}\) or the pass limit is reached}
\end{algorithmic}
\end{algorithm}

\subsection{TPCCAP-SA}

TPCCAP-SA adds a simulated-annealing stage after TPCCAP. A candidate move from \(\pi\) to \(\pi'\) has objective difference
\begin{equation}
\Delta=J(\pi')-J(\pi).
\end{equation}
The move is always accepted when \(\Delta\leq0\). Otherwise, it may be accepted with probability
\begin{equation}
P_{\mathrm{accept}}=\exp\left(-\frac{\Delta}{T}\right),
\end{equation}
where \(T\) is the current temperature. This permits occasional uphill moves and can escape local minima. The method remains heuristic and does not guarantee a globally optimal partition.

\begin{algorithm}[htbp]
\caption{TPCCAP-SA refinement}
\label{alg:sa}
\begin{algorithmic}[1]
\Require Feasible partition \(\pi\), objective \(J\), temperature schedule
\Ensure Best partition found during the run
\State Set \(\pi_{\mathrm{best}}\leftarrow\pi\).
\For{each annealing step}
  \State Propose a capacity-preserving move or swap.
  \State Compute \(\Delta=J(\pi')-J(\pi)\).
  \If{\(\Delta\leq0\) or a uniform random draw is below \(\exp(-\Delta/T)\)}
    \State Accept \(\pi'\).
    \If{\(J(\pi)<J(\pi_{\mathrm{best}})\)}
      \State Update \(\pi_{\mathrm{best}}\leftarrow\pi\).
    \EndIf
  \EndIf
  \State Update the temperature.
\EndFor
\State Return \(\pi_{\mathrm{best}}\).
\end{algorithmic}
\end{algorithm}

\subsection{Communication-port selection}

After partitioning, \QuPort{} selects which logical qubits should occupy communication-qubit positions. For a logical qubit \(i\), define its external score
\begin{equation}
s_i=\sum_{j:\pi(j)\neq\pi(i)}w_{ij}.
\end{equation}
The top-k mode selects the highest external-score qubits in each QPU. The diverse mode also considers the remote QPUs contacted by the candidate boundary qubits, so that selected communication qubits are not all focused on the same remote neighbor when alternatives exist.

\begin{algorithm}[htbp]
\caption{Communication-port-aware layout}
\label{alg:layout}
\begin{algorithmic}[1]
\Require Partition \(\pi\), edge weights \(w\), architecture blocks
\Ensure Physical layout \(\ell\)
\For{each logical qubit \(i\)}
  \State Compute external score \(s_i\).
\EndFor
\For{each QPU \(q\)}
  \State Select up to \(P\) logical qubits assigned to \(q\) for communication positions.
  \State Map selected logical qubits to communication physical qubits of \(q\).
  \State Map remaining logical qubits to compute physical qubits of \(q\).
\EndFor
\State Reject incomplete or overflowing layouts.
\end{algorithmic}
\end{algorithm}

\section{Distributed Program Construction}

In distributed mode, \QuPort{} applies the partition-aware layout without allowing a global inter-QPU routing pass to hide remote interactions. The mapped circuit is scanned in instruction order. A one-qubit operation is appended to the local circuit of the QPU that owns the operand. A two-qubit operation is appended locally only if both operands belong to the same QPU. If the operands belong to different QPUs, the operation is recorded as a remote event and synchronization barriers are inserted into the affected local circuits.\\

The resulting program is
\begin{equation}
\mathcal{D}(C)=\left(\{C_q\}_{q\in V_Q},\mathcal{R}\right),
\end{equation}
where \(C_q\) is the local circuit for QPU \(q\), and \(\mathcal{R}\) is the ordered list of remote operations. Each remote operation records the operation name, physical operands, endpoint QPUs, parameters, classical bits, and original instruction index. Local circuits are later routed using only the intra-QPU coupling map for their own QPU.

\begin{algorithm}[htbp]
\caption{Distributed program extraction}
\label{alg:split}
\begin{algorithmic}[1]
\Require Mapped physical circuit, modular architecture
\Ensure Local circuits \(\{C_q\}\) and remote-event list \(\mathcal{R}\)
\State Initialize one local circuit for each QPU.
\For{each instruction in mapped-circuit order}
  \State Determine the physical operands and their owning QPUs.
  \If{the instruction has no quantum operands}
    \State Append it where applicable, or propagate barriers to local circuits.
  \ElsIf{all operands belong to one QPU}
    \State Append the instruction to that QPU's local circuit.
  \ElsIf{the instruction is a two-qubit cross-QPU operation}
    \State Append a remote-operation record to \(\mathcal{R}\).
    \State Insert synchronization barriers on the participating local circuits.
  \Else
    \State Treat the cross-QPU multi-qubit instruction conservatively as a remote composite event.
  \EndIf
\EndFor
\end{algorithmic}
\end{algorithm}

This intermediate representation is deliberately not a claim about a specific physical remote-gate protocol. It exposes where such a protocol must be supplied by a backend.

\section{Scheduling and Abstract Cost Semantics}

\QuPort{} includes a topology-aware estimator for comparing compiler choices under fixed abstract parameters. The estimator uses circuit layers as a dependency-aware approximation. Local one-qubit, two-qubit, and SWAP instructions contribute abstract costs \(\tau_1\), \(\tau_2\), and \(\tau_{\mathrm{swap}}\). For remote operations, the model uses three abstract terms: entanglement generation cost \(\tau_E\), classical round-trip cost \(\tau_C\), and remote-operation overhead \(\tau_R\). If asynchronous classical overlap is enabled, the effective classical term is
\begin{equation}
\tau_C^{\mathrm{eff}}=(1-\rho)\tau_C,
\qquad 0\leq\rho\leq1.
\end{equation}
For a remote operation between QPUs \(a\) and \(b\), the topology-aware remote cost is modeled as
\begin{equation}
\tau_{\mathrm{remote}}(a,b)=d(a,b)\tau_E+\tau_C^{\mathrm{eff}}+\tau_R.
\end{equation}

Within a circuit layer, remote operations are packed greedily into rounds. An operation can be placed in a round only if both endpoint QPUs have available communication ports and each interconnect link on the selected shortest path has remaining link capacity. For switch-like topologies, the estimator can also account for a limit on distinct QPU pairs per round and an optional reconfiguration delay. The output contains makespan, number of layers, number of remote operations, number of remote rounds, peak link utilization, and peak QPU-port usage.

\begin{figure}[htbp]
\centering
\begin{tikzpicture}[
  qpu/.style={circle, draw, minimum size=7mm, inner sep=0pt, font=\scriptsize},
  box/.style={draw, rounded corners, align=center, minimum height=8mm, minimum width=30mm, font=\scriptsize},
  edge/.style={thick},
  used/.style={very thick},
  arr/.style={-{Latex[length=2mm]}, thick}
]
\node[qpu] (a) at (0,0.8) {$a$};
\node[qpu] (b) at (2.0,0.8) {$b$};
\node[qpu] (c) at (4.0,0.8) {$c$};
\node[qpu] (d) at (2.0,-0.8) {$d$};
\draw[used] (a) -- node[above, font=\tiny] {used link} (b);
\draw[used] (b) -- (c);
\draw[edge] (a) -- (d);
\draw[edge] (b) -- (d);
\draw[edge] (c) -- (d);
\node[box] (round) at (7.0,1.0) {remote round};
\node[box, below=6mm of round] (ports) {endpoint port count\(\leq P\)};
\node[box, below=6mm of ports] (links) {path link usage\(\leq\) capacity};
\draw[arr] (c) -- (round);
\end{tikzpicture}
\caption{Remote-round feasibility in the topology-aware estimator. A round must respect endpoint communication-port limits and link-capacity limits along the chosen QPU-network paths.}
\label{fig:scheduler}
\end{figure}
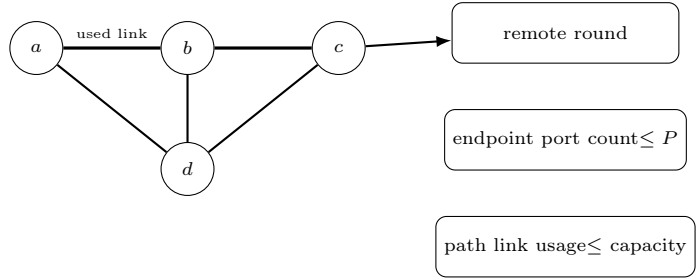

The scalar cost model used in the global compilation path is also abstract:
\begin{equation}
C_{\mathrm{local}}=\tau_1n_1+\tau_2n_2+\tau_{\mathrm{swap}}n_{\mathrm{swap}},
\end{equation}
\begin{equation}
C_{\mathrm{remote}}=n_{\mathrm{remote}}(\tau_E+\tau_C+\tau_R),
\end{equation}
\begin{equation}
C_{\mathrm{total}}=C_{\mathrm{local}}+C_{\mathrm{remote}}+0.1d_C\tau_2,
\end{equation}
where \(d_C\) is circuit depth. These equations are useful only under a fixed set of abstract parameters.

\section{Implementation Properties}

The implementation is organized around the same conceptual separation used in the mathematical model. The configuration module defines the modular architecture and latency parameters. The architecture module builds local and inter-QPU coupling structures. The network module computes QPU graphs, shortest paths, traffic matrices, routed link loads, and congestion metrics. The interaction module extracts logical two-qubit weights. The partitioning module implements heavy-edge clustering, balanced greedy partitioning, TPCCAP, and TPCCAP-SA. The layout module assigns selected boundary logical qubits to communication qubits. The distributed module extracts local circuits and remote events. The scheduler estimates topology-aware remote rounds and makespan. The pipeline and compiler modules connect these pieces to Qiskit transpilation.\\

Several correctness conditions are enforced by validation. Partition inputs must have valid logical indices and finite nonnegative weights. Total logical demand must not exceed total capacity. Layout construction rejects invalid QPU assignments and incomplete physical layouts. QPU shortest-path data are validated before use in TPCCAP. Traffic matrices and link-load maps are checked for shape, symmetry, nonnegative values, and finite entries.

\begin{proposition}
Every partition returned by heavy-edge clustering, balanced greedy partitioning, TPCCAP, or TPCCAP-SA satisfies the per-QPU capacity constraint.
\end{proposition}
\begin{proof}
Heavy-edge clustering merges components only when the merged size does not exceed capacity and places components only into QPUs with sufficient remaining space. Balanced greedy placement considers only non-full QPUs. Its local refinement moves a qubit only to a destination QPU with available capacity. TPCCAP starts from a feasible balanced partition and applies only capacity-preserving moves. TPCCAP-SA proposes only capacity-preserving moves or swaps. Thus the capacity invariant is preserved by every state transition.
\end{proof}

\begin{proposition}
For a connected source-destination QPU pair, the equal-cost multi-path routing routine conserves the injected traffic weight.
\end{proposition}
\begin{proof}
The routine constructs the shortest-path directed acyclic graph for the source and destination. Let \(\sigma(v)\) be the number of shortest paths from the source to vertex \(v\). During backward accumulation, flow at vertex \(v\) is split among predecessor vertices in proportion to \(\sigma(u)/\sigma(v)\). Since \(\sigma(v)\) is the sum of \(\sigma(u)\) over all shortest-path predecessors, the outgoing shares from \(v\) sum to the incoming flow at \(v\). Applying this argument layer by layer from the destination to the source proves that the total returned to the source equals the injected flow.
\end{proof}

\section{Compiler Semantics and Limitations}

The framework should be interpreted as a compiler framework for modular-mapping studies under abstract latency assumptions. It does not provide a calibrated backend for a trapped-ion, superconducting, photonic, or neutral-atom platform. It also does not implement the physical operation that realizes a remote event. The remote-event list is an intermediate representation that identifies where such an operation would be required.\\

The algorithms are heuristic. Heavy-edge clustering, balanced greedy partitioning, TPCCAP, and TPCCAP-SA are designed to expose and reduce compiler-level resource pressure; they do not guarantee optimal graph partitions. The schedule estimator is a layer-based greedy model; it is not a verified network-control scheduler. These boundaries are part of the artifact's semantics and are necessary for correct interpretation of its outputs.

\section{Conclusion}

\QuPort{} provides a concrete framework for modular quantum compilation in which logical placement, communication-port assignment, interconnect topology, congestion, remote-event extraction, and local-only routing are represented explicitly. Its central contribution is not a hardware-specific remote-gate implementation, but a compiler abstraction that keeps local routing separate from QPU-level communication. TPCCAP uses this abstraction to optimize weighted cut distance, communication-port overflow, and routed congestion. The distributed compilation path turns cross-QPU gates into ordered remote events while preserving local circuits for Qiskit routing. This makes \QuPort{} suitable for studying the algorithmic structure of modular compilation before committing to a specific physical interconnect protocol.

\bibliographystyle{splncs04}
\bibliography{references}

\appendix
\renewcommand{\theHsection}{appendix.\arabic{section}}

\section{Notation}

\begin{table}[htbp]
\centering
\caption{Notation used in the manuscript.}
\label{tab:notation}
\begin{tabular}{@{}ll@{}}
\toprule
Symbol & Meaning \\
\midrule
\(Q_L\) & logical-qubit set \\
\(G_L\) & weighted logical interaction graph \\
\(G_Q\) & undirected QPU-level interconnect graph \\
\(C\) & compute qubits per QPU \\
\(P\) & communication qubits per QPU \\
\(B=C+P\) & physical block size of one QPU \\
\(\pi\) & logical-to-QPU partition \\
\(\ell\) & logical-to-physical layout \\
\(T_{ab}\) & traffic between QPUs \(a\) and \(b\) \\
\(L_e\) & routed load on QPU-network edge \(e\) \\
\(J(\pi)\) & TPCCAP objective \\
\bottomrule
\end{tabular}
\end{table}

\end{document}